\documentclass[10pt]{article}
\usepackage{amssymb}
\newcommand{\bbbone}{{\mathchoice {\rm 1\mskip-4mu l} {\rm 1\mskip-4mu l}    {\rm 1\mskip-4.5mu l} {\rm 1\mskip-5mu l}}}

\newcommand{\bC}{{\mathbb C}}

\newcommand{\bN}{{\mathbb N}}

\newcommand{\bX}{{\mathbb X}}

\newcommand{\de}{\delta}

\newcommand{\et}{\eta}

\renewcommand{\th}{\theta} 

\newcommand{\rh}{\rho}
\newcommand{\si}{\sigma}

\newcommand{\ps}{\psi}

\newcommand{\vphi}{\varphi}

\newcommand{\Ps}{\Psi}

\newcommand{\bGa}{{{\rm I}\kern-.16em \Gamma}}
\newcommand{\beq}{{\bar\beta}}

\newcommand{\etq}{{\bar\eta}}

\newcommand{\thq}{{\bar\theta}}

\newcommand{\psq}{{\bar\psi}}

\newcommand{\cC}{{\cal C}}
\newcommand{\cD}{{\cal D}}
\newcommand{\cE}{{\cal E}}
\newcommand{\cF}{{\cal F}}
\newcommand{\cG}{{\cal G}}
\newcommand{\cH}{{\cal H}}

\newcommand{\cK}{{\cal K}}

\newcommand{\cN}{{\cal N}}

\newcommand{\cV}{{\cal V}}
\newcommand{\cW}{{\cal W}}

\newcommand{\del}{\partial}

\newcommand{\abs}[1]{{\left\vert #1 \right\vert}}
\newcommand{\norm}[1]{{\left\Vert #1 \right\Vert}}

\newcommand{\Ref}[1]{$(\ref{#1})$}
\newcommand{\True}[1]{\; \bbbone\left( #1 \right) \;}
\newcommand{\sfrac}[2]{{\textstyle \frac{#1}{#2}}}

\newcommand{\pli}{\prod\limits}

\newif\ifintrmk
\newcommand{\E}{{\rm e}}
\newcommand{\I}{{\rm i}}
\newcommand{\dd}{{\rm d}}
\renewcommand{\beq}{\begin{equation}}
\newcommand{\eeq}{\end{equation}}

\newcommand{\cre}{a^+}
\newcommand{\ann}{a^-}
\newcommand{\psp}{{\psi_+}}
\newcommand{\psm}{{\ps_-}}

\newcommand{\Gra}[1]{{\cG}({#1})}
\newcommand{\Kin}{\cE}

\newcommand{\Tr}{\mbox{ Tr }}

\newcommand{\UN}{{u_N}}

\newcommand{\gsi}{c^+}
\newcommand{\gsiq}{c^-}
\newcommand{\ggsi}{c}

\newcommand{\Ind}{X} 
\newcommand{\Xx}{x}
\newcommand{\Xy}{y}

\newcommand{\tnorm}[1]{|\!|\!| #1 |\!|\!|}

\newtheorem{satz}{Theorem}

\newtheorem{lemma}[satz]{Lemma}

\newenvironment{proof}{\par\noindent {\it Proof:} \hspace{7pt}}%
{\hfill\hbox{\vrule width 7pt depth 0pt height 7pt} \par\vspace{10pt}\noindent}
\begin{document}

\title{Clustering of fermionic truncated expectation values
via functional integration}
\author{Manfred Salmhofer \\ 
\small Institut f\" ur Theoretische Physik, Universit\" at Heidelberg \\ 
\small Philosophenweg 19, 69120 Heidelberg, Germany}

\date{}
\maketitle

\normalsize 
\begin{abstract}
\noindent
I give a simple proof that the correlation functions of many-fermion systems
have a convergent functional Grassmann integral representation, 
and use this representation to show that the 
cumulants of fermionic quantum statistical 
mechanics satisfy $\ell^1$--clustering estimates.
\end{abstract}

\section{Introduction}
Formal coherent--state functional integrals of the type 
\begin{eqnarray}
Z
&=& 
\mbox{ Tr } \E^{ - \beta (H - \mu N)}
\\
&=&
\int 
\prod\limits_{\tau\in [0,\beta)\atop x \in \Lambda}
\dd \bar\phi(\tau,x) \dd \phi(\tau,x)\;
\E^{- \int_0^\beta \dd \tau 
\left( 
\bar\phi(\tau,x) ({\del \over \del \tau} + \mu ) \phi(\tau,x)
- H (\bar\phi(\tau) , \phi(\tau))
\right)}
\nonumber
\end{eqnarray}
have become ubiquitous in the physics literature. These formulas, 
and their generalizations to correlation functions, have become 
powerful heuristic and calculational tools in theoretical physics. 
However, even when space $\Lambda$ is replaced by a {\em finite} set, 
their mathematical meaning is not obvious because the product measure 
over a continuum of $\tau $ does not exist. 
The natural way to obtain a mathematically well-defined identity 
is to use the Lie--Trotter product formula to introduce a discrete 
Euclidian time $\tau$ and recover a rigorous variant of the above formal 
equation in the limit where the time discretization $\Delta \tau$ vanishes.
This involves the solution of a mild, but not completely trivial, 
ultraviolet problem, which arises because only a first--order derivative
with respect to $\tau$ appears in the exponent, so that
the covariance of the Gaussian integral
on the right hand side is not absolutely summable as a function of the 
frequency $\omega$ dual to $\tau$. This problem has been solved both for 
fermions and bosons. The solution for bosons is very recent \cite{BFKT}; 
for fermions, a solution that does not require rather obscure 
multiscale arguments was first given in \cite{FKT} at zero temperature 
($\beta \to \infty$) 
and then in \cite{PSUV} for all $\beta > 0$ (the basic problems and the 
background are discussed in \cite{PSUV}). 

The main technical difficulty in the study of condensed matter systems is, 
of course, not this ultraviolet problem, but the infrared problem associated
to perturbations of Hamiltonians that do not have a gap above the ground state, 
entailing a slow, non--summable, decay of the above--mentioned Gaussian covariance 
at long times and distances, and much effort has gone into understanding its
physical and mathematical consequences.  
On the other hand, proving simple statements about the reduced density matrices
of equilibrium states requires a solution of the UV problem as well. This may be
the reason why no simple functional integral  
proof of decay properties of truncated correlations (a.k.a cumulants)
seems to have been published. 

In this paper I prove the convergence of the above functional integral representation 
for fermions on a finite lattice. The bounds for the generating function for the 
reduced density matrices will be uniform in the lattice size, so that the thermodynamic limit 
can be taken. 

This paper is partly motivated by the work of J.\ Lukkarinen and H.\ Spohn, 
who control the many-body time evolution 
on the short kinetic timescale \cite{LS}
under the hypothesis of an $\ell^1$ clustering property of the 
truncated expectation values. 
The result of \cite{LS} only requires estimates at some 
fixed inverse temperature $\beta$, but not sharp bounds on the dependence on 
$\beta$, so that, using the results of \cite{PSUV}, multiscale techniques 
can be avoided altogether. 
Inspection of the hypothesis of \cite{LS} reveals that the clustering property
needed there is a special case of finiteness of the natural norms used 
in the analysis of fermionic systems in \cite{FKT} and \cite{SW,PSUV}. 
Thus the results given here are easy applications of the results of 
\cite{SW} and \cite{PSUV}, and the main purpose of this paper is therefore 
to make them more widely accessible. To this end, I also include in the Appendix 
a largely streamlined derivation of the functional integral representation 
along the lines of \cite{msbook}, Appendix B, and a new proof
of its convergence using the bounds of \cite{PSUV}.

\section{Setup and main results}

Let $\Ind$ be a finite set. In typical applications, 
$\Ind = \Lambda \times \{1, \ldots, n\}$, where
$\Lambda$ is a spatial lattice and the second index
labels all internal degrees of freedom, e.g.\ spin.
This choice is not necessary; one could also think of 
$\Ind$ as indexing any orthonormal basis of a finite--dimensional
Hilbert space $\cH$, such as the energy eigenbasis of a system
in a trapping potential (with an ultraviolet cutoff). 
Let $\cF = \bigwedge \cH$ be the 
antisymmetric Fock space over $\cH$; it is also finite--dimensional.
Let $\cre $ and $\ann$ denote the standard fermionic creation
and annihilation operators satisfying canonical anticommutation 
relations.

Let $\rho$ be a density operator of the form 
$\rho = Z^{-1} \E^{-\beta (H - \mu N)}$ where $H= H_0 + V$
and $N$ is the number operator (we also call $H_0 - \mu N = K_0$). 
On the finite lattice, all these operators are bounded, and
the partition function $Z = Z_\Ind(\beta,\mu)= $ Tr$_\cF \E^{-\beta (H - \mu N)}$
is well--defined.

\begin{equation}
\langle  A \rangle_{\beta,\mu,\Ind}
=
\frac{1}{Z_\Ind(\beta,\mu)}
\Tr \left(
\E^{-\beta (H - \mu N)} \; A 
\right) 
\end{equation}
(provided that $Z_\Ind \ne 0$, which will be the case in our applications,
where $H=H^*$).

\subsection{Normal--ordered monomials}
 
The reduced density matrices 
in the grand canonical state associated to $H$, $\beta$ and $\mu$ are
\begin{equation}
\gamma_{m,n} (\Xx_1,\ldots,\Xx_m;\Xy_1,\ldots ,\Xy_n)
=
\left\langle 
\pli_{i=1}^m \cre_{\Xx_i} \pli_{j=n}^1 \ann_{\Xy_j}
\right\rangle_{\beta,\mu,\Ind} .
\end{equation}
If the Hamiltonian is $U(1)$--charge--invariant, $\gamma_{m,n} \ne 0$ only if 
$m=n$, in which case I denote it by $\gamma_m$.

The reduced density matrices of the state can be obtained by taking 
derivatives of the following generating functional. 
Let $(\gsi_\Xx)_{\Xx \in \Ind}$ and $(\gsiq)_{\Xx \in \Ind}$ 
be Grassmann source fields satisfying for all $\Xx,\Xy \in \Ind$ and all 
$u,v \in \{-1,1\}$
\begin{equation}
\ggsi_\Xx^u \ggsi_\Xy^v + \ggsi_\Xy^v \ggsi_\Xx^u =0, 
\quad
\ggsi_{\Xx}^u a^{v}_\Xy + a^{v}_\Xy \ggsi_{\Xx}^u = 0
\end{equation}
for all $\Xx,\Xy \in \Ind$, and define
\begin{equation}\label{eq:genred}
Z(\gsiq,\gsi) 
= 
\Tr 
\left\lbrack
\E^{-\beta (H - \mu N)} \; 
\E^{(\gsi, \cre)_\Ind} \E^{(\gsiq, \ann)_\Ind}
\right\rbrack
\end{equation}
where $(\gsi,\cre)_\Ind = \int_\Ind \dd \Xx \; \gsi_x \cre_x$
(here we use a continuum notation $\int_\Ind \dd \Xx = \varepsilon \sum_{\Xx \in \Ind}$ where $\varepsilon$ can be a suitable power of a lattice spacing).
Then 
\begin{equation}
\gamma_m (\Xx_1,\ldots,\Xx_m;
\Xy_1,\ldots ,\Xy_m)
=
\sfrac{\de^m}{\de\gsi_{\Xx_1} \ldots \de \gsi_{\Xx_m}} \;
\sfrac{\de^m}{\de\gsiq_{\Xy_m} \ldots \de \gsiq_{\Xy_1}} \;
\sfrac{Z(\gsiq,\gsi)}{Z(0,0)} \Big\vert_{\gsi=0,\gsiq=0}
\end{equation}
with $\frac{\de}{\de \ggsi^{\pm}_x} = \varepsilon^{-1} \frac{\del}{\del \ggsi^{\pm}_x}$. 
As will become clear in the following, it is natural to take Grassmann-valued
sources.

\subsection{General monomials}
In applications, one often wants to get information about the expectation 
values of monomials where annihilation operators can be to the left of 
creation operators. This easily included in the present setup as follows.
Although the two exponentials under the trace in \Ref{eq:genred}
do not commute, they are easily 
combined into a single one, as follows. 
Because the $\gsi$ and $\gsiq$ are Grassmann variables, 
\begin{equation}
[\gsi_{\Xx}\cre_\Xx, \gsiq_{\Xy} \ann_\Xy]
=
\gsi_{\Xx} \gsiq_{\Xy} 
\{ \cre_\Xx,\ann_\Xy \}
= 
\gsi_{\Xx} \gsiq_{\Xy}\de_{\Ind}(\Xx,\Xy)
\end{equation}
which commutes both with $(\gsi,\cre)_\Ind $ and $(\gsiq,\ann)_\Ind$.
The Baker--Campbell--Hausdorff formula then implies that 
\begin{eqnarray}
\E^{(\gsi,\cre)_\Ind} \; \E^{(\gsiq, \ann)_\Ind}
&=&
\E^{(\gsi,\cre)_\Ind+(\gsiq, \ann)_\Ind + \frac12 (\gsi,\gsiq)_\Ind}
\nonumber\\
&=&
\E^{(\gsiq, \ann)_\Ind}\; \E^{(\gsi,\cre)_\Ind} \; \E^{(\gsi,\gsiq)_\Ind}
\label{eq:commu}
\end{eqnarray}
Thus commuting two such exponentials of source terms produces 
an additional factor $\E^{-(\gsiq , \gsi)_\Ind}$.

Consider now the expectation value of a monomial
$\pli_{i=1}^{2m} a^{\si_i} (\Xx_i)$, where 
$\si = (\si_1, \ldots \si_{2m} ) \in \{ -1,1\}^{2m}$ 
is an arbitrary sequence of creation/annihilation indices.
For $i \in \{1, \ldots, 2m\}$ with $\si_i =+$ let 
$L_i (\si) = \{ j < i : \si_j = -1\}$ be the set of 
labels of annihilation operators to the left of $i$; 
if $\si_i = -1$, set $L_i = \emptyset$. 
Then by \Ref{eq:commu}
\begin{eqnarray}\label{reord}
\pli_{i=1}^{2m} \E^{ \ggsi^{\si_i}_{i} a^{\si_i} (\Xx_i)}
=
\E^{- \sum_i \sum_{j \in L_i (\si)} (\ggsi_j^-,\ggsi_i^+)_\Ind} \;
: \pli_{i=1}^{2m} \E^{ \ggsi^{\si_i}_{i} a^{\si_i} (\Xx_i)} :
\end{eqnarray}
It is an elementary exercise to show that for quasifree states, 
the prefactor gives the familiar $f_\beta$, $1-f_\beta$ structure for the 
expectation values of unordered monomials in quasifree states
(used in the context of deriving the quantum Boltzmann equation in
\cite{ESYQBE} and \cite{LS}). Here $f_\beta$ is the Fermi function
(see below).

\subsection{Truncated expectations}\label{allthesame}

The truncated expectations (also called cumulants, or 
connected parts, of the reduced density matrices), are  
the derivatives of the logarithm of $Z(\gsiq,\gsi)$:
let
\begin{equation}
F(\gsiq,\gsi) = \log  Z(\gsiq,\gsi) 
\end{equation}
then
\begin{equation}
\gamma_m^T (\Xx_1, ... ,\Xx_m; \Xx_{m+1}, ... , \Xx_{2m}) 
=
\sfrac{\de^m}{\de\gsi_{\Xx_1} ... \de \gsi_{\Xx_m}} \;
\sfrac{\de^m}{\de\gsiq_{\Xx_{2m}} \ldots \de \gsiq_{\Xx_{m+1}}} \;
F (\gsiq,\gsi) \Big\vert_{\gsi=0,\gsiq=0}.
\end{equation}
By straightforward expansion of the exponential in $Z = \E^{F}$
and evaluation of the Grassmann derivatives at zero sources 
$\gsiq$ and $\gsi$, one can verify that the $\gamma_m^T$ 
indeed coincide with the truncated expectations, as defined
in \cite{BR}.

Equation \Ref{reord} makes evident that for $m \ge 2$, 
the thus defined truncated expectations $\gamma_m^T$ are indeed 
totally antisymmetric under permutations of $\Xx_1, \ldots, \Xx_{2m}\;$
(as stated in Theorem A.1 of \cite{LS}): 
when taking the expectation value of \Ref{reord}
and then the logarithm, the $\sigma$--dependent prefactor
gives a quadratic contribution which drops out in the derivatives
with respect to the sources $\gsiq$ and $\gsi$ for $m > 1$. 
Antisymmetry follows because the product of exponentials is
invariant under permutations, and because Grassmann derivatives 
anticommute. Moreover, \Ref{reord} implies that the 
truncated expectation values of order $m \ge 2$ of the sequence of monomials
$\left(\pli_{i=1}^{2m} a^{\si_i} (\Xx_i)\right)_{m \in \bN}$ 
are in absolute value 
independent of the ordering, therefore bounds on the
truncated reduced $m$--particle density matrices
imply the same bounds for all truncated expectations 
(for $m \ge 2$).

\subsection{$\ell^1$--clustering}
Let $H_0$ be selfadjoint so that 
$K_0 = H_0 - \mu N = (\cre, \Kin \ann)_\Ind$ 
with a hermitian matrix $\Kin$. 
Let 
\beq
\bC (\tau, E ) 
= 
- 1_{\tau > 0} f_\beta (-E) \E^{-\tau E} 
+ 
1_{\tau \le 0} f_\beta (E) \E^{-\tau E}
\eeq
where 
$f_\beta (E) = (1+ \E^{\beta E})^{-1}$ is the Fermi function,
and let
\begin{equation}\label{timeodef}
\cC_{\tau,\tau'} (\cE)
=
\bC (\tau-\tau',\cE) .
\end{equation}
be the standard Euclidian--time--ordered free fermion two--point function.
$\cC$ is an $\bX \times \bX$--matrix with $\bX = [0,\beta) \times \Ind$.
Let $\delta > 0$ be a determinant bound for $\cC$, as defined in 
Definition 1.2 of \cite{PSUV},
namely $\delta$ is such that for all $n \in \bN$  and all
$x_1, \dots, x_n, y_1, \dots, y_n \in \Ind $ 
\begin{equation}\label{eq:detbound}
 \sup\limits_{p_1, \ldots, p_n, q_1, \ldots, q_n \in B}
\abs{\det \left(\langle p_i \, , \, q_j \rangle C_{x_i y_j}\right)_{1 \le i,j \le n}}
\le {\delta}^{2n}
\end{equation}
(here $B = \{ \xi \in \bC^n : \norm{\xi}_2 \le 1\}$),
and let
\begin{equation}\label{eq:alpha}
\alpha 
=
\sup\limits_{\Xy \in \Ind}
\int_0^\beta \dd \tau
\int_\Ind \dd \Xx 
\abs{\cC (\tau,\cE)}_{\Xx,\Xy}
\end{equation}
be the decay constant of $\cC$ (see Section 4.2 of \cite{PSUV}).
Both of these constants are finite for any $\beta >0$ and finite 
index set $\Ind$; $\delta$ is uniform in $\beta $ and in $\Ind$ if
$f_\beta (\pm \cE)$ is trace class uniformly in $\Ind$ (see below), 
and $\alpha$ is uniform in $\beta$ if $\cC$ decays fast enough
(for sufficient conditions for this in terms of $\cE$,
see \cite{PSUV} and below).

The interaction $V$ can be a generic charge--invariant multibody interaction, 
assumed to be given in normal ordered form
\begin{equation}
V 
=
\sum_{m \in \bN} 
\int \pli_{i=1}^{2m} \dd \Xx_i \;
v_{m}(\Xx_1, \ldots, \Xx_{2m})\;
\cre_{\Xx_1} \ldots \cre_{\Xx_m} \ann_{\Xx_{m+1}} \ldots \ann_{\Xx_{2m}}
\end{equation}
For $f_m: \Ind^{2m} \to \bC $, let 
\beq\label{1infdef}
\abs{f_m}_{1,\infty}
=
\max\limits_j
\sup\limits_{\Xx_j}
\int \pli_{i=1 \atop i \ne j}^{2m} \dd \Xx_i \;
\abs{f_{m}(\Xx_1, \ldots, \Xx_{2m})} .
\eeq
For $h >0$ and a sequence $f = (f_m)_{m \in \bN}$ let
\beq\label{hnormdef}
\norm{f}_h 
=
\sum_{m \ge 1} \abs{f_{m}}_{1,\infty} h^{2m}
\eeq

\begin{satz}\label{th:clu}
Let $\omega = 2 \alpha \delta^{-2}$, and assume that the interaction
$V$ has the properties $V=V^*$ and $\omega \norm{V}_{3\delta} \le 1/2$. 
Then 
for all $m \ge 2$, $\gamma^T_m$ obeys $\ell^1$--clustering,
i.e.
\beq\label{ell1clu}
\abs{\gamma^T_m}_{1,\infty}
\le
2 (m!)^2\; \alpha^{2m} \; \delta^{-2m} \; \norm{V}_{3\delta} , 
\eeq
and the same bound holds for truncated expectation values of 
unordered monomials of degree $2m \ge 4$.
For $m=1$ (and denoting $\gamma_1^T$ at $V=0$ by $\gamma_{1,0}^T$)
\beq\label{2pf}
\abs{\gamma_1^T - \gamma_{1,0}^T}_{1,\infty}
\le 
2 \alpha^{2} \delta^{-2} \norm{V}_{3\delta} .
\eeq 
\end{satz}

\begin{proof}
Eqs.\ 
\Ref{ell1clu} and \Ref{2pf} are proven  in Section \ref{FIRsec}.
That the truncated expectation values of 
unordered monomials of degree $2m \ge 4$ 
satisfy the same bounds as $\gamma^T_m$ 
was shown in Section \ref{allthesame}.
\end{proof}

\bigskip\noindent
{\bf Remarks.}

\begin{enumerate}

\item
For a two--body interaction $V = \int_\Ind \dd \Xx \int_\Ind \dd \Xy
:n_\Xx \; v(\Xx , \Xy) n_\Xy:$, 
$\norm{V}_h = h^4 \abs{v}_{1,\infty}$ is small for small (and summable)
$v$. Thus, when measuring the strength of the interaction, 
one can regard the coefficient of $h^4$ as the coupling constant,
$\lambda = \abs{v}_{1,\infty}$. For general $V$, 
I take the convention that, up to the lowest power of $h$ that appears
in it, $\norm{V}_h$ is the ``coupling constant''. With this convention, 

\begin{itemize}
\item
Equation \Ref{ell1clu} implies (6.19) of \cite{LS}
because $(m!)^2 \le (2m)!$. 

\item
Equation \Ref{2pf} implies (6.20) of \cite{LS}. 

\end{itemize}

The factorial in  \Ref{ell1clu} reflects the 
property that the radius of analyticity of 
$F(\gsiq,\gsi)$ as a function of $(\gsiq,\gsi)$
is finite, while the partition function
$Z(\gsiq,\gsi)$ itself
is an exponentially bounded entire function
of $(\gsiq,\gsi)$.
The advantage of using truncated functions is, of course, that the radius of 
convergence for $\gamma_m^T$ is uniform in $\abs{\Ind}$ and in $N$, 
and that the expansion coefficients converge in the limits
$\abs{\Ind} \to \infty$ and $N \to \infty$, 
while the expansion coefficients in $Z$ diverge
for $\abs{\Ind} \to \infty$.

\item
Similar bounds hold for states that are not charge invariant; 
the restriction was imposed here mainly for notational simplicity.

\item
Estimates for the determinant bound and the decay constant 
are given in \cite{PSUV} for the translation--invariant case
where $\Lambda$ is a discrete torus of sidelength $L$, so that 
the propagator $\cC$ has a standard Fourier representation.
In particular, in this case, $\delta$ is uniform in $L$. 
The decay constant $\alpha $ will in general depend on $\beta$,
but there are a few cases where it does not, in particular
if there is an energy gap, 
i.e. the distance of $0$ to the spectrum of $\cE$ is bounded
below by a fixed positive number.

\item 
More generally, let $(\vphi_\Xx)_{\Xx \in \Ind}$ 
be an orthonormal basis of $\cH$, and let 
\beq
C_{(\tau,\Xx),(\tau',\Xx')} 
= 
\langle \vphi_\Xx | \cC(\tau-\tau',\cE) \vphi_{\Xx'} \rangle .
\eeq 
The spectral representation 
$\cE = \sum_l \epsilon_l |e_l \rangle \langle e_l |$
then gives a natural Gram representation (with $\E^{-\I p \cdot x}$
replaced by $\langle e_l \mid \vphi_{\Xx} \rangle$ and integration 
over $p$ replaced by summation over $l$) that 
replaces the Fourier representation in Lemma 4.1 of \cite{PSUV}.
Proceeding as in the proof of this Lemma, one obtains
\beq
\delta 
\le
2 \max\limits_{\sigma=\pm} \sum_l f_\beta (\sigma \epsilon_l)
\eeq
which is bounded uniformly in $\Ind$ if $f_\beta (\pm \cE)$ is trace 
class uniformly in $\Ind$. This effectively requires an ultraviolet
cutoff because $f_\beta (-E) \to 1 $ as $E \to \infty$. 
The decay constant depends on the properties of the $e_l$ and, 
again, on the presence or absence of spectrum of $\cE$ near $E=0$. 

\item
More detailed bounds can be obtained by the methods 
described in Section \ref{FIRsec}: by using weighted norms, one can
obtain precise decay estimates; by using more details of 
tree expansions \cite{SW} one can also obtain pointwise estimates.
All these estimates depend on the decay constant and the determinant
bound of the covariance. In multiscale approaches to the low--temperature
infrared problem, the bounds given in Section \ref{FIRsec} recover 
the full power counting, as already noted in \cite{PSUV}.

\end{enumerate}

\section{Functional integral representation}\label{FIRsec}
To analyse connected functions, it is convenient to have a graded algebra 
in which all even elements commute. This is never the case in the CAR algebra 
of fermionic operators on Fock space. This is the main motivation for using 
a Grassmann integral representation of the grand canonical traces.

$Z(\gsiq, \gsi )$ has the following Grassmann integral representation. 

\begin{satz}\label{th:Gras}
Let $H_0$ be selfadjoint so that 
$K_0 = H_0 - \mu N = (\cre, \Kin \ann)_\Ind$ 
with a hermitian matrix $\Kin$. 
Let $V$ be any linear map from $\cF$ to $\cF$.

For $m \in \{1,\ldots, N\}$, let $\tau_m = m \frac{\beta}{N}$
and $R^{(N)}$ be the operator on $\bC^N \otimes \bC^{|\Ind|}$
given by
\beq\label{eq:RNdef}
R^{(N)}_{m,n}
=
\cC_{\tau_n, \tau_m} (\Kin) 
\eeq
with $\cC $ given by \Ref{timeodef}.  
Then
\begin{equation}\label{e7}
Z(\gsiq, \gsi ) 
= 
Z_0\;
\E^{(\gsiq, f_\beta(\Kin) \gsi)_\Ind} \;
\lim\limits_{N \to \infty}
\left(
\mu_{R^{(N)}} * \E^{-\cV}
\right) 
(\etq^{(N)} , \et^{(N)} )
\end{equation}
where for $k \in \{ 1, \ldots, N\}$
\beq
\etq_k^{(N)} = \cC_{\beta, \tau_k}^{\rm T} \gsiq,
\quad
\et_k^{(N)} = \cC_{\tau_k,\tau_1} \gsi ,
\eeq
(the {\rm T} denotes the transpose as an operator on $\bC^{|\Ind|}$), 
$Z_0 = \det (1 + \E^{-\beta \Kin}) = $ Tr $\E^{-\beta(H_0 -\mu N)}$
is the partition function for $V=0$, and
\begin{equation}
\cV= \frac{\beta}{N} \sum_{i=1}^N \cN (V) (\psq_i,\ps_i)
\end{equation}
Here $\cN (V)$ is the normal ordered form of $V$, 
and $*$ denotes the standard Grassmann Gaussian convolution 
integral (see \cite{msbook})
\beq
(\mu_C * F) (\bar\phi,\phi)
=
\int \dd \mu_C (\psq,\ps) F (\psq+\bar\phi, \ps+\phi)
\eeq
where the integration is over the Grassmann variables 
$(\psq_{k,\Xx},\ps_{k,\Xx})_{(k,\Xx) \in \{1,\ldots,N \}\times \Ind}$.
\end{satz}

\bigskip\noindent
Theorem \ref{th:Gras} is proven in the appendix. 

\bigskip\noindent
In the Theorem, we recognize 
$\E^{-\cW^{(N)}(R^{(N)},\cV)} = \mu_{R^{(N)}} * \E^{-\cV}$
as {\em Wilson's effective action}, which generates the connected amputated 
Green function of the time--discretized
theory. In \cite{SW} and \cite{PSUV}, this function was shown 
to converge uniformly in $\Ind$, provided that $\delta$ and
$\alpha$ are uniform in $\Ind$. The variables $\eta^{(N)}$ and $\et^{(N)}$
on which $\cW$ depends contain a $\cC$ because $Z$ 
generates non--amputated functions, where a propagator
is associated to each external leg. 
The sources $\ggsi$ appear in the time slices $\tau_1$ and $\tau_N=\beta$, 
as expected since the monomial ``sits'' at the fixed time $\tau =0$. 

I briefly recall some of the properties of $\cW$ proven in \cite{SW,PSUV}.
By Theorem 4.5 of \cite{PSUV}, the limit 
$\cW (\cC,\cV)
= 
\lim_{N \to \infty} \cW^{(N)} (R^{(N)},\cV) 
$
exists in $\norm{\cdot}_h$ for small enough $h$.
Let $h=\delta$ in that theorem and denote
the part of $\cW$ that is homogeneous of degree $p$ in $\cV$ 
by $\cW (\cC, \cV; p)$, then
\beq\label{normleq}
\norm{
\cW(\cC, \cV)
-
\sum_{p=1}^P
\frac{1}{p!}
\cW (\cC, \cV; p)}_{\delta, \bX}
\le
\frac{(\omega \norm{\cV}_{3\delta,\bX})^P}{1- \omega \norm{\cV}_{3\delta,\bX}} 
\norm{\cV}_{3\delta,\bX}
\eeq 
holds for all $P \ge 0$, provided $\omega \norm{\cV}_{3\delta,\bX} < 1$
(for $P=0$ the empty sum in \Ref{normleq} is zero; 
this is the only case needed to prove Theorem \ref{th:clu}).

The norm $\norm{\cdot}_{h,\bX} $ is similar to the norm
$\norm{\cdot}_{h} $ defined in \Ref{hnormdef}, 
but with $\abs{\cdot}_{1,\infty}$ in \Ref{1infdef} 
replaced by a $1,\infty$--norm $\abs{\cdot}_{1,\infty,\bX}$
in which integrals and suprema run over 
$\bX = [0,\beta) \times \Ind$.
Because the interaction $V$ gives rise to a $\cV$ that is local in 
the time $\tau \in [0,\beta)$, 
$\norm{\cV}_{h,\bX} = \norm{V}_h$.

In summary, for small enough $\omega \norm{V}_{3\delta}$, 
$\cW$ is analytic in $\cV$ and in the fields.
In particular, $\cW$ is continuous in the fields,
and
\beq
\lim_{N \to \infty} \cW^{(N)} (R^{(N)},\cV) (\etq^{(N)},\eta^{(N)})
=
\cW (\cC, \cV ) (\etq,\eta)
\eeq
with $\etq (\tau) = \cC^{\rm T}_{\beta,\tau} \gsiq$ and 
$\et (\tau) = \cC_{\tau,0} \gsi $.
Let $\cW^{(N)}_{2m}$ be the homogeneous part of $\cW^{(N)}$ of degree $2m$
in $\etq,\eta$. It generates the connected amputated $2m$--point functions
$W^{(N)}_{2m} = \frac{1}{(m!)^2}\frac{\de^m}{\de \eta^m}\,\frac{\de^m}{\de \etq^m} \cW^{(N)}_{2m}$. 
Then $W^{(N)}_{2m}$ is given by a tree expansion
\cite{SW,PSUV}, i.e.\ an explicit, 
absolutely and uniformly in $N$ convergent expansion in which 
the limit $N \to \infty$ can be taken termwise. 
Denote the limit by $W_{2m}$. Then, if 
$\omega \norm{V}_{3\delta} \le 1/2$,
\beq\label{Wmleq}
\abs{W_{2m}}_{1,\infty, \bX}
\le 
2 \delta^{-2m} \; \norm{V}_{3 \delta}
\eeq
for all $m \ge 1$, and
\beq
Z(\gsiq, \gsi ) 
= 
Z_0\;
\E^{(\gsiq, f_\beta(\Kin) \gsi)_\Ind} \;
\E^{-\cW (\etq,\eta)}
\eeq
so $F= \log Z$ becomes
\begin{eqnarray}\label{full2}
F
(\gsiq, \gsi ) 
&=& 
\log Z_0 + (\gsiq, f_\beta(\Kin) \gsi)_\Ind
-\cW (\etq,\et) 
\\
&=&
\log Z_0 + (\gsiq, 
\left[f_\beta(\Kin) 
- 
\cK_\beta
\right]\gsi)_\Ind
-\sum_{m \ge 2} \cW_{2m}(\cC_{\beta,\,\cdot\,} \gsiq,\cC_{\,\cdot\, , 0} \gsi)
\nonumber
\end{eqnarray}
with $\cK_\beta = \int \dd \tau \int \dd \tau' \; \cC_{\beta,\tau}
\cW_2 (\tau,\tau') \cC_{\tau',0} $.
In \Ref{full2}, the full two-point function is explicit in the quadratic part 
of the generating function. The analogue of the free fermion formula, 
obtained by discarding the $\cW_{2m}$ for all $m \ge 2$, 
changes correspondingly.
Under the present assumptions, $\cW$ is small when
$\norm{\cV}_h$ is small, and hence it is clear that all higher truncated
functions are small and $\cK_\beta$ is small in $\abs{\cdot}_{1,\infty}$, 
i.e.\ the two--point function ($m=1$) is close to the free one.

Theorem \ref{th:clu} now follows immediately from the above properties of $\cW$:
Because the arguments of $\cW$ in \Ref{full2} are 
$\etq = \cC_{\beta,\,\cdot\,} \gsiq$ and $\eta = \cC_{\,\cdot\, , 0} \gsi$, 
\begin{eqnarray}
\gamma_m^T (\Xx_1, \ldots, \Xx_{2m}) 
&=&
\int \pli_{k=1}^{2m} \dd \tau_k \dd \Xy_k \; 
(m!)^2
W_{2m} \left(
(\tau_1,\Xy_1), \ldots, (\tau_{2m},\Xy_{2m})
\right)
\nonumber\\
&&
\pli_{k=1}^m 
\left(\cC (\beta,\tau_k)\right)_{\Xx_k,\Xy_k}
\left(\cC (\tau_{k+m},0)\right)_{\Xy_{k+m},\Xx_{k+m}} 
\end{eqnarray}
for all $m \ge 2$, hence
\beq
\sfrac{1}{(m!)^2}
\abs{\gamma_m^{T}}_{1,\infty}
\le
\abs{W_{2m}}_{1,\infty,\bX}\;
\abs{\cC}_{1,\infty,\bX}^{2m}
=
\abs{W_{2m}}_{1,\infty,\bX}\;
\alpha^{2m}
\eeq
by \Ref{eq:alpha}. By \Ref{Wmleq}, 
\beq
\abs{\gamma_m^{T}} 
\le
2 (m!)^2\;\delta^{-2m} \; \norm{V}_{3 \delta} \alpha^{2m}
\eeq
which proves \Ref{ell1clu}. 
For $m=1$, recall that $\gamma^T_{1,0} = f_\beta (\cE)$, so 
$\gamma_{1}^T - \gamma^T_{1,0} = \cK_\beta$, as defined after
\Ref{full2}. Thus, again by \Ref{Wmleq} and \Ref{eq:alpha}
\beq
\abs{\gamma_{1}^T - \gamma^T_{1,0}}_{1,\infty}
=
\abs{\cK_\beta}_{1,\infty}
\le 
\alpha^2 \abs{W_2}_{1,\infty,\bX}
\le
2 \alpha^2 \delta^{-2} \norm{V}_{3\delta}, 
\eeq
which proves \Ref{2pf}.

\appendix
\section{Traces and Grassmann Integrals}

\begin{lemma}\label{le:liemod}
Let $\cF$ be a Hilbert space and   $A$ and $B$ be
bounded linear operators on $\cF$. 
\begin{equation}
\E^{A+B}
=
\lim\limits_{N \to \infty}
\left[
\E^{\frac{A}{N} } \; (1 + \frac{B}{N} )
\right]^N
\end{equation}
\end{lemma}

\begin{proof}
This is an easy variant of the standard proof of the Lie product formula, 
as given, e.g.\ in \cite{msbook}. Let $C= \E^{(A+B)/N}$ and 
$D=\E^{A/N} (1+ B/N)$. Then 
\begin{equation}
\max\{\norm{C}, \norm{D} \}
\le 
\E^{\frac{\norm{A}+\norm{B}}{N}}
\end{equation}
and by following the proof given in \cite{msbook}, one obtains the estimate
\begin{equation}
\norm{C^N - D^N} 
\le
\frac{4}{N}
(\norm{A}+\norm{B})^2 \;
\E^{\norm{A}+\norm{B}} .
\end{equation}
\end{proof}

\noindent
By assumption, the one--particle Hilbert space is finite--dimensional, 
i.e.\ has an ONB indexed by a finite set $\Ind$. 
Thus the Fock space $\cF$ is finite--dimensional, too, 
hence all linear operators on $\cF$ are bounded 
and the trace is a continuous linear functional. 
Therefore by Lemma \ref{le:liemod}, 
\begin{equation}
Z(\gsiq, \gsi ) 
= 
\lim\limits_{N \to \infty}
\tilde Z_N (\gsiq,\gsi)
\end{equation}
where 
\begin{eqnarray}
\tilde Z_N (\gsiq,\gsi)
=
\Tr 
\left\lbrack
\rho_N \;
\E^{(\gsi, a_+)} \E^{(\gsiq, a)}
\right\rbrack
\end{eqnarray}
and 
$\rho_N = (C_ND_N)^N$ 
with $C_N=\E^{-\frac{\beta}{N} K_0 }$ and $D_N=1 - \sfrac{\beta}{N} V $.
If $V$ is normal ordered, the same is true for $D_N$ (but not for 
$\E^{-\beta V/N}$). This is the main reason that makes Lemma \ref{le:liemod}
convenient in the transition to the Grassmann integral done in the 
following. 

By a purely algebraic transformation,
$\tilde Z_N (\gsiq,\gsi)$ is represented as 
a finite--dimensional Grassmann integral. 
Some basic facts needed for this are in Appendix B of \cite{msbook}.
Here I collect only the most important definitions and identities.
The  {\em Grassmann symbol} $\Gra{A}$ of an operator $A$
with normal ordered form $\cN (A) (\cre,\ann)$
is defined as
\begin{equation}\label{eq:Grasdef}
\Gra{A} (\psp,\psm) 
=
\E^{(\psp,\psm)_\Ind} \; \cN (A) (\psp,\psm) .
\end{equation}
The product of two operators has Grassmann symbol
\begin{equation}\label{eq:prosym}
\Gra{AB}(\psq,\ps)
=
\int \cD_\Ind (\psq',\ps') \;
\Gra{A}(\psq,\ps')\;
\E^{-(\psq',\ps')_\Ind }\;
\Gra{B}(\psq',\ps)
\end{equation}
The trace is represented by 
\begin{equation}
\Tr A
=
\int \cD_\Ind (\psp,\psm) \;
\Gra{A}(-\psp,\psm)\;
\E^{-(\psp,\psm)_\Ind }
\end{equation}
By these identities, $Z$ can be written as the convolution
\begin{equation}\label{eq:Zsisi}
\tilde Z_N (\gsiq,\gsi)
=
\int 
\cD_\Ind (\thq,\th) \;
\Gra{\rh_N}
(\bar\th,\th)
\;
\E^{(\bar\th+\gsiq,\th+\gsi)} .
\end{equation}
By the rule for products of Grassmann symbols, \Ref{eq:prosym}, 
\begin{eqnarray}
\tilde Z_N (\gsiq,\gsi)
&=&
\int 
\pli_{k=1}^N 
\cD_\Ind (\psq_k,\ps_k) \;
\E^{(\psq_1 + \gsiq, \ps_N + \gsi)_\Ind}
\nonumber\\
&&
\pli_{k=1}^N
\Gra{C_N D_N} (\psq_k,\ps_k) \;
\pli_{k=1}^{N-1}
\E^{-(\psq_{k+1},\ps_k)_\Ind}\;
\end{eqnarray}
The decomposition $H=H_0 + V$ is chosen such that $H_0$ is quadratic, 
so that $K_0 = (\cre , \Kin \ann)_\Ind$ with a hermitian matrix $\Kin$. 
For this special case, the Grassmann symbol is easily obtained by 
Grassmann Gaussian integration, with the result
\begin{equation}
\Gra{C_N D_N} (\psq_k,\ps_k)
=
\Gra{D_N} (\UN^T \psq_k,\ps_k)
\end{equation}
where $\UN = \E^{-\frac{\beta}{N}\Kin}$ (and $T$ denotes the transpose).
The interaction $V$ is assumed to be in normal ordered form, 
so its Grassmann symbol is straightforwardly obtained by replacing
$\cre$ by $\psq$ and $\ann$ by $\psi$. 
Inserting this and changing variables from $\psq_k$ to 
$\UN^T \psq_k$ gives
\begin{eqnarray}
\tilde Z_N (\gsiq,\si)
&=&
\det \E^{-\Kin}
\int 
\pli_{k=1}^N 
\cD_\Ind (\psq_k,\ps_k) \;
\E^{(\UN^T \psq_1 + \gsiq, \ps_N + \gsi)_\Ind}
\nonumber\\
&&
\pli_{k=1}^n
\left(
1 - \frac{\beta}{N} V(\psq_k,\ps_k)
\right) \;
\pli_{k=1}^{N-1}
\E^{-(\UN^T \psq_{k+1},\ps_k)_\Ind}\;
\nonumber\\
&=&
\E^{(\gsiq,\gsi)_\Ind} \; 
\det \E^{-\Kin}
\int 
\pli_{k=1}^N 
\cD_\Ind (\psq_k,\ps_k) \;
\E^{(\psq, Q^{(N)} \ps )}
\nonumber\\
&&
\pli_{k=1}^n
\left(
1 - \frac{\beta}{N} V(\psq_k,\ps_k)
\right) \;
\E^{(\gsiq, \ps_N)_\Ind + (\psq_1 , \UN \gsi)_\Ind}
\end{eqnarray}
with 
\begin{equation}
(\psq, Q^{(N)} \ps )
=
\sum_{k=1}^N
(\psq_k, \ps_k)_\Ind 
-
\sum_{k=1}^{N-1}
(\psq_{k+1}, \UN \ps_k)_\Ind 
+ 
(\psq_1, \UN \ps_N)_\Ind
\end{equation}
That $\psq_1$ couples to $\ps_N$ with the opposite sign
means that there is an antiperiodic boundary condition for
the fields $\psq$ and $\ps$.
$Q^{(N)}$ acts as a matrix on the ``time'' indices $k$ and as an 
operator on the $\Xx \in \Ind$:
\begin{equation}
Q^{(N)}_{m,n}
=
\de_{m,n}
-
\UN \de_{m-1,n}
+
\UN \de_{m,1} \de_{N,n}
=
Q^{(0,N)}_{m,n} + \UN \de_{m,1} \de_{N,n}
\end{equation}
$Q^{(0,N)}$ is lower triangular and its inverse is easily found 
to be $R^{(0,N)}$
\begin{equation}
R^{(0,N)}_{kl}
=
\UN^{k-l} \True{k \ge l}
\end{equation}
The additional term in $Q$ is similar to a rank one perturbation. 
Using the standard formulas, one obtains $(Q^{(N)})^{-1} = R^{(N)}$
as 
\begin{eqnarray}
R^{(N)}_{m,n}
&=&
\UN^{m-n} 
\left(
\True{m \ge n} 
-
(1 + \UN^{-N})^{-1}
\right)
\nonumber\\
&=&
\UN^{m-n}
\cases{
- (1 + \UN^{-N})^{-1} & for $m < n$\cr
(1 + \UN^{N})^{-1} & for $m \ge n$
}
\end{eqnarray}
we obtain the formula \Ref{eq:RNdef} for $R^{(N)}$.
Thus, in summary 
\begin{eqnarray}
\tilde Z_N (\gsiq,\gsi)
&=&
\E^{(\gsiq,\gsi)_\Ind} \; 
Z_0^{(N)}
\int 
\dd\mu_{R^{(N)}} (\psq,\ps) \; 
\E^{(\gsiq, \ps_N)_\Ind + (\psq_1 , \E^{-\frac{\beta}{N} \Kin} \gsi)_\Ind}
\nonumber\\
&&
\pli_{k=1}^n
\left(
1 - \frac{\beta}{N} V(\psq_k,\ps_k)
\right) \;
\end{eqnarray}
where $Z_0^{(N)} = \det \E^{-\Kin} \; \det Q^{(N)}$.

It remains to reexponentiate $V$. 
To this end, recall from Definition B.11 of \cite{msbook} the 
$\ell^1$ seminorm $\tnorm{\cdot}_q$ 
on the Grassmann algebra,
namely $\tnorm{A}_q = \sum_{m \ge 1} \abs{A_m}_1 q^m$
with $\abs{A_m}_1 = \int \prod_{i=1}^{m} \dd \Xx_i \;
\abs{A_{m}(\Xx_1, \ldots, \Xx_{m})}$.
It satisfies the product inequality
$\tnorm{AB}_q \le \tnorm{A}_q \tnorm{B}_q$, 
and if the covariance $C$ has determinant bound $\delta$, 
then 
\begin{equation}\label{Gintnorm}
\abs{\int \dd \mu_C(\Psi) F(\Psi) }
\le
\tnorm{F}_{\delta} 
\end{equation}
(see Theorem B.14 in \cite{msbook}).
Now let 
\begin{equation}
\Delta (\Psi)
=
\pli_{k=1}^N \E^{-\frac{\beta}{N} V (\Ps_k)}
-
\pli_{k=1}^N 
\left(
1 - \frac{\beta}{N} V (\Ps_k) 
\right) .
\end{equation}
By the discrete product rule
\begin{equation}
\pli_{k=1}^N A_k
-
\pli_{k=1}^N B_k
=
\sum_{l=1}^N
\pli_{k < l} A_k\;
(A_l-B_l) \;
\pli_{k>l} B_k ,
\end{equation}
and by \Ref{Gintnorm} and the triangle and product inequalities, 
\begin{eqnarray}
\abs{\int \dd \mu_C(\Psi) \Delta (\Psi) }
&\le&
\tnorm{\Delta}_\delta
\nonumber\\
&\le &
N 
\E^{\frac{(N-1)}{N} \beta \tnorm{V}_\delta}
\tnorm{\E^{-\frac{\beta}{N} V} - 1 + \frac{\beta}{N} V }_\delta
\nonumber\\
&\le&
N 
\E^{\frac{(N-1)}{N} \beta \tnorm{V}_\delta}
\frac{\beta^2 \tnorm{V}_q^2}{N^2} \;
\E^{\frac{\beta}{N} \tnorm{V}_\delta}
\nonumber\\
&=&
\frac{\beta^2 \tnorm{V}_q^2}{N} \;
\E^{\beta \tnorm{V}_\delta}\to 0
\end{eqnarray}
as $N \to \infty$. 
Note that if $\Ind$ is a lattice of volume $\Omega$
with a translation group, $\tnorm{G}_q$ is typically
of order $\Omega$ for translation invariant elements $G$ of the 
Grassmann algebra, but this does not matter here because
$\Omega$ remains fixed as $N \to \infty$. 

Thus
\begin{equation}
Z(\gsiq, \gsi ) 
= 
\lim\limits_{N \to \infty}
Z_N (\gsiq,\gsi)
\end{equation}
where
\begin{eqnarray}
Z_N (\gsiq,\gsi)
&=&
\E^{(\gsiq,\gsi)_\Ind} \; 
\det (1+ \E^{-\Kin}) \; 
\int 
\dd\mu_{R^{(N)}} (\psq,\ps) \; 
\E^{(\gsiq, \ps_N)_\Ind + (\psq_1 , 
\gsi)_\Ind}
\nonumber\\
&&
\E^{-\frac{\beta}{N} \sum_{k=1}^n V(\psq_k,\ps_k)}
\end{eqnarray}
with $R^{(N)}$ given in \Ref{eq:RNdef}.
Eq.\ \Ref{e7} now follows by a standard completion of the square
in the integrand. 

\bigskip\noindent
{\bf Acknowledgement. }
This work was supported by DFG grants FOR 718 and FOR 723.
Part of it was done during a very inspiring summer school 
at the {\em Erwin--Schr\" odinger--Institut} in Vienna. 
I would like to thank Christian Hainzl, Robert Seiringer, 
and Jakob Yngvason for the invitation to this summer school,
and Herbert Spohn for discussions.

\end{document}